\documentclass[letterpaper]{llncs}

\usepackage{soulutf8} 
\usepackage{todonotes} 
\usepackage{cite} 
\usepackage{amssymb}
\usepackage{amsmath}
\usepackage{textcomp}
\usepackage{array}
\usepackage{color}
\usepackage{setspace}
\usepackage{comment}
\usepackage[]{datetime}
\usepackage{hyperref}

\newcommand{\bra}[1]{\langle #1|}
\newcommand{\ket}[1]{|#1\rangle}
\newcommand{\braket}[2]{\langle #1|#2\rangle}

\newcommand{\cent}[0]{\mbox{\textcent}}
\newcommand{\dollar}[0]{\$}

\newcommand{\mymatrix}[2]{\left( \begin{array}{#1} #2\end{array} \right)}

\newcommand{\myvector}[1]{\mymatrix{c}{#1}}

\newcommand{\tildesigma}{\widetilde{\Sigma}}
\newcommand{\tildew}{\tilde{w}}

\newcommand{\reg}{\mathsf{REG}}
\newcommand{\stoc}{\mathsf{SL}}
\newcommand{\affine}{\mathsf{AfL}}
\newcommand{\naffine}{\mathsf{NAfL}}
\newcommand{\baffine}{\mathsf{BAfL}}
\newcommand{\posonebaffine}{\mathsf{BAfL^0}}
\newcommand{\negonebaffine}{\mathsf{BAfL^1}}
\newcommand{\exactaffine}{\mathsf{EAfL}}
\newcommand{\excstoc}{\mathsf{SL^{\neq}}}
\newcommand{\coexcstoc}{\mathsf{SL^{=}}}
\newcommand{\excstocrational}{\mathsf{SL^{\neq}_\mathbb{Q}}}
\newcommand{\coexcstocrational}{\mathsf{SL^{=}_\mathbb{Q}}}
\newcommand{\excaffine}{\mathsf{AfL^{\neq}}}
\newcommand{\coexcaffine}{\mathsf{AfL^{=}}}
\newcommand{\nqal}{\mathsf{NQAL}}

\newcommand{\TSigma}{\widetilde{\Sigma}}

\newcommand{\modtwo}{\mathtt{MOD2^k}}
\newcommand{\modtwoyes}{\mathtt{0MOD2^k}}
\newcommand{\modtwono}{\mathtt{1MOD2^k}}
\newcommand{\modfour}{\mathtt{MOD4^k}}
\newcommand{\modfouryes}{\mathtt{0MOD4^k}}
\newcommand{\modfourno}{\mathtt{1MOD4^k}}
	
\newcommand{\interval}[2]{\mathtt{INTERVAL}_{#1,#2}}

\newcommand{\E}{\mathtt{E}}
\newcommand{\emptylang}{\mathtt{\varnothing}}
\newcommand{\less}[1]{\mathtt{LESS}_{#1}}
\newcommand{\even}{\mathtt{EVEN}}
	
\title{Language Recognition Power and Succinctness of Affine Automata\thanks{A preliminary and shorter version will appear in the proceedings of the 15th International Conference on Unconventional Computation and Natural Computation (UCNC 2016) \cite{VY16B}.}}

\author{Marcos Villagra\inst{1} \and Abuzer Yakary{\i}lmaz\inst{2}$^,$\thanks{Yakary{\i}lmaz was partially supported by CAPES with grant 88881.030338/2013-01 and some parts of this work was done while he was visiting Universidad Nacional de Asunci\'on in September 2015.}}
\institute{
  Universidad Nacional de Asunci\'on\\
  NIDTEC, Campus Universitario, San Lorenzo, C.P. 2619, Paraguay\\
  \email{mvillagra@pol.una.py}
  \and 
  National Laboratory for Scientific Computing\\
  Petr\'{o}polis, RJ, 25651-075, Brazil\\
  \email{abuzer@lncc.br}
}

\authorrunning{M. Villagra \and A. Yakary{\i}lmaz} 

\begin{document}

\maketitle

\begin{abstract}
In this work we study a non-linear generalization based on affine transformations of probabilistic and quantum automata proposed recently by D\'iaz-Caro and Yakary{\i}lmaz \cite{DCY16A} referred as affine automata. First, we present efficient simulations of probabilistic and quantum automata by means of affine automata which allows us to characterize the class of exclusive stochastic languages. Then, we initiate a study on the succintness of affine automata. In particular, we show that an infinite family of unary regular languages can be recognized by 2-state affine automata but the state numbers of quantum and probabilistic automata cannot be bounded. Finally, we present the characterization of all (regular) unary languages recognized by two-state affine automata.
\end{abstract}

\section{Introduction}

\subsection{Background}
Probabilistic and quantum computing are popular computation models with a very rich literature. Quantum computation, in particular, apparently violates the so-called \emph{strong Church-Turing thesis}, which states that all reasonable models of computation can be efficiently simulated by a probabilistic universal Turing machine. Evidence comes from the efficient solution to certain problems believed to be computationally hard, like factoring large composite numbers. Much research is devoted to pinpoint the exact source of this computational power of quantum computers. 

In this paper, we continue the work initiated in \cite{DCY16A} on a quantum-like classical computational model based on affine transformations. In particular, we make emphasis in finite-state automata, which is arguably the most simple computation model. Affine automata are finite-state machines whose transition operators are affine operators, hence the name.

There are several sources that apparently gives power to quantum computers, like quantum parallelism and entanglement. Several researchers may agree that quantum interference (using negative amplitudes), however, seems to be the key component. Therefore, the reason to study affine automata is to simplify the study of quantum interference in the context of a simple classical computation model.

Probabilistic automata are computation models whose transitions are governed by stochastic operators preserving the $\ell_1$-norm of a normalized vector with entries in the continuous set of real numbers $[0,1]$. Similarly, the transitions in a quantum automaton are governed by unitary operators preserving the $\ell_2$-norm of a normalized vector with entries over the complex numbers $\mathbb{C}$. The only restriction that affine transformations impose over finite-state machines is the preservation of barycenters of vectors with entries over the real numbers $\mathbb{R}$, or equivalently, preservation of the sum of all entries in a state vector. It is clear that any affine operator defined on non-negative real numbers is a stochastic operator. 

Since affine transformations are linear, the evolution of an affine automaton is linear.
Nonlinearity comes from a measurement-like operator (which we call weighting operator) that is applied at the end of every computation to determine the probability of observing an inner-state of the machine. We refer the reader to \cite{DCY16A} for the detailed explanations and discussions. A continuation of this paper appeared in \cite{BMY16A}.

\subsection{Contributions}
In this work we present the following results on affine automata language classes. First, in Section \ref{sec:pfa-sim} we show how to simulate a probabilistic automaton using an affine automaton (Theorem \ref{the:pfa-sim}). Then we use that simulation to show that any rational exclusive stochastic language can be recognized by positive one-sided bounded-error affine automata (Theorem \ref{the:sl-bafl}). This fact immediately implies a characterization of the language recognition power of nondeterministic quantum automata by one-sided bounded-error affine automata. In Section \ref{sec:mcqfa-sim} we show how to simulate an $n$-state quantum automaton exactly by an $(n^2+1)$-state affine automaton (Theorem \ref{the:qfa-sim}). 
In Section \ref{sec:succ} we study the state complexity (succintness) of affine automata. First, we show that the so-called unary counting problem can be computed by some bounded-error affine automata with constant state complexity (Theorem \ref{the:count}), whereas any bounded-error quantum automaton requires at least a logarithmic number of states. Second, we show the existence of a promise language that is solved exactly by an affine automaton with constant state complexity (Theorem \ref{the:mod-afa}), whereas any probabilistic automaton requires exponential state complexity. Finally, in Section \ref{sec:unary} we give a complete characterization of all (regular) unary languages recognized by two-state affine automata (Theorem \ref{the:cha}).

Affine transformations are arguably simpler to understand compared to unitary operators. Therefore, the characterizations given in terms of affine automata of quantum language classes present a simpler setting where to study and research the power of interference.


\section{Preliminaries}\label{sec:preliminaries}
We assume the reader is familiar with the common notation used in automata theory. For details on the models of probabilistic and quantum automata, we recommend references \cite{Paz71}, \cite{SayY14}, and \cite{AY15}. 

Let $ \Sigma $ be a finite alphabet, not containing $ \cent $ and $ \dollar $ called the left and right end-markers, respectively. The set of all the strings of finite length over $ \Sigma $ is denoted $ \Sigma^* $. We define $ \tildesigma = \Sigma \cup \{ \cent,\dollar \} $ and $ \tildew = \cent w \dollar $ for any string $w \in \Sigma^* $. For any given string $ w \in \Sigma^* $, $|w|$ denotes its length, $|w|_\sigma$ is the number of occurrences of the symbol $\sigma$, and $w_j$ is the $j$-th symbol of $w$. 


A \emph{probabilistic finite automaton} (or PFA) \cite{Rab63} is a 5-tuple
$P = (E,\Sigma,\{ A_{\sigma} \mid \sigma \in \tildesigma \},e_s,E_a)$,
where $ E = \{ e_1,\ldots,e_n \} $ is a finite set of inner states for some $ n \in \mathbb{Z}^+ $, $e_s \in E$ is the starting inner state, $ E_a \subseteq E $ is a set of accept inner states, and $ A_{\sigma} $ is the stochastic transition matrix for the symbol $ \sigma \in \widetilde{\Sigma} $. Any input $ w \in \Sigma^* $ is always given in the form $\tilde{w}=\cent w\$$ and it is scanned by $P$ from left to right, symbol by symbol.\footnote{This way of scanning an input tape is sometimes referred to as ``strict realtime.''} After scanning the $ j $-th symbol, the configuration state of $P$ is
$v_j = A_{\tildew_j} v_{j-1} = A_{\tildew_j} A_{\tildew_{j-1}} \cdots A_{\tildew_1} v_0$,
where $ 1 \leq j \leq |\tildew| $ and $ v_0 $ is the $s$-th element of the standard basis in $ \mathbb{R}^n $, referring to the initial state. The final configuration state is denoted $ v_f = v_{|\tildew|} $. The acceptance probability of $ P $ on $w$ is given by
$f_P(w) = \sum_{e_k \in E_a} v_f[k]$,
where $v_f[k]$ is the $k$-th entry of the vector $v_f$.

A \emph{quantum finite automaton} (or QFA) \cite{AY15} is a 5-tuple
$M = (Q,\Sigma,\{ \mathcal{E}_\sigma \mid \sigma \in \TSigma \},q_s,Q_a)$,
where $ Q = \{ q_1,\ldots,q_n \} $ is a finite set of inner states for some $ n \in \mathbb{Z}^+ $, $ \mathcal{E}_\sigma $ is a transition superoperator\footnote{A superoperator or quantum operator is a positive-semidefinite operation that maps density matrices to density matrices \cite{SayY14,AY15}.}  for a symbol $ \sigma\in \Sigma $, the inner state $ q_s $ is the initial state, and $ Q_a \subseteq Q $ is a set of accept states. For any given input $ w \in \Sigma^* $, the computation of $ M $ on $ w $ is given by
$\rho_j = \mathcal{E}_{\tilde{w}_j} ( \rho_{j-1} )$,
where $ \rho_0 = \ket{q_s}\bra{q_s} $ and $ 1 \leq j \leq |\tilde{w}| $. The final state is denoted $ \rho_f = \rho_{|\tilde{w}|} $. The accept probability of $ M $ on $ w $ is given by
$f_M(w) =  \sum_{q_j \in Q_a} \rho_f[j,j]$,
where $\rho_f[j,j]$ is the $j$-th diagonal entry of $\rho$. The most restricted model of QFA currently known is the so-called \emph{Moore-Crutchfield QFA} (or MCQFA) \cite{MC00}. An MCQFA is a 5-tuple
$M = (Q,\Sigma,\{ U_\sigma \mid \sigma \in \TSigma \},q_s,Q_a)$,
where all components are defined exactly in the same way as for QFAs except that $ U_\sigma $ is a  unitary transition operator for a symbol $ \sigma \in \Sigma$ acting on $span\{\ket{q} \mid q\in Q\}$. Physically, $M$ corresponds to a closed-system based on pure states.\footnote{Pures states are vectors in a complex Hilbert space normalized with respect to the $\ell_2$-norm.} For any given input $ w \in \Sigma^* $, the machine $M$ is initialized in the quantum state $ \ket{v_0} = \ket{q_s} $. Each step of a computation is given by $\ket{v_j} = U_{\tildew_j} \ket{v_{j-1}}$,
where $ 1 \leq j \leq |\tilde{w}| $. The final quantum state is denoted $ \ket{v_f} = \ket{v_{|\tilde{w}|}} $. The accept probability of $ M $ on $ w $ is
$f_M(w) =  \sum_{q_j \in Q_a} |\braket{q_j}{v_f}|^2$.
Note that the inner product $ \braket{q_j}{v_f} $ gives the amplitude of $ q_j $ in $ \ket{v_f} $.

If we restrict the entries in the transitions matrices of a PFA to zeros and ones we obtain a \emph{deterministic finite automaton} (or DFA). A DFA is always in a single inner state during the computation and the input is accepted if only if the computation ends in an accept state. A language is said to be recognized by a DFA if and only if any member of the language is accepted by the DFA. Any language recognized by a DFA is called a \emph{regular language} \cite{RS59} and the class of regular languages is denoted $\reg$. 

Let $ \lambda \in \lbrack 0,1) $ be a real number. A language $ L $ is said to be recognized by a PFA $P$ with \emph{cutpoint} $ \lambda $ if and only if
$L = \{ w \in \Sigma^* \mid f_P(w) > \lambda \}$.
Any language recognized by a PFA with a cutpoint is called a \emph{stochastic language} \cite{Rab63} and the class of stochastic languages is denoted $ \stoc $, which is a superset of $\reg$.  As a special case,  if $ \lambda = 0 $, the PFA is also called a \emph{nondeterministic finite automaton} (or NFA). Any language recognized by an NFA is also a regular language. 

A language $ L $ is said to be recognized by $ P $ with \emph{isolated cutpoint} $ \lambda $ if and only if there exists a positive real number $ \delta $ such that 
$ f_P(w) \geq \lambda + \delta $ for any $ w \in L $ and
$ f_P(w) \leq \lambda - \delta $ for any $ w \notin L $. When the cutpoint is required to be isolated, PFAs are not more powerful than DFAs; that is, any language recognized by a PFA with isolated cutpoint is regular \cite{Rab63}. 

Language recognition with isolated cutpoint can also be formulated as recognition with bounded error. Let $ \epsilon \in \lbrack 0,\frac{1}{2}) $. A language $ L $ is said to be recognized by a PFA $ P $ with error bound $ \epsilon $ if and only if $ f_P(w) \geq 1 - \epsilon $ for any $ w \in L $ and $ f_P(w) \leq \epsilon $ for any $ w \notin L $.
  
As a further restriction, if $ f_P(w) = 1 $ for any $ w \in L $, then we say that $P$ recognizes $L$ with \emph{negative one-sided bounded error}; if $ f_P(w) = 0 $ for any $ w \notin L $, then we say that $P$ recognizes $L$ with \emph{positive one-sided bounded error}. If the error bound is not specified, then we say that $ L $ is recognized by $ P $ with \emph{\lbrack negative/positive one-sided\rbrack\ bounded error}. 

A language $ L $ is an \emph{exclusive stochastic language} \cite{Paz71} if and only if there exists a PFA $P$  and a cutpoint $ \lambda \in [0,1]  $ such that 
$L = \{ w \in \Sigma^* \mid f_P(w) \neq \lambda \}$.
The class of exclusive stochastic languages is denoted by $ \excstoc $. Its complement class is denoted by $ \coexcstoc $ (that is $ L \in \excstoc$ iff $\overline{L}\in \coexcstoc $). Note that for any language in $ \excstoc $ we can choose as cutpoint any value between 0 and 1, but not 0 or 1,  because in that case  we can only recognize regular languages. Also notice that both $\excstoc$  and $ \coexcstoc$ are supersets of $ \reg $ (it is still open whether $ \reg $ is a proper subset of $ \excstoc \cap \coexcstoc $). 

In the case of QFAs, they recognize all and only regular languages with bounded-error \cite{LQZLWM12} and stochastic languages with cutpoint \cite{YS09C,YS11A}. The class of languages recognized by nondeterministic QFAs, however, is identical to $ \excstoc $. 

For any language class $ \tt C $, we use $ \mathtt{C}_{\mathbb{X}} $ to denote the subclass of $ \tt C $ when all transitions of the corresponding model are restricted to $ \mathbb{X} $.

\section{Affine Finite Automaton}\label{sec:affine}

In this section we define our model of finite-state machine based on affine transformations. We refer \cite{DCY16A} for the basics of affine systems.  An \emph{affine finite-state automaton}, or simply AfA,  is a 5-tuple
\[
  M = (E,\Sigma,\{ A_{\sigma} \mid \sigma \in \tildesigma\},e_s,E_a),
\]
where all the components are the same as in the definition of a PFA excepting that $ A_{\sigma} $ is an affine transformation matrix (each column sum is 1). Note that each configuration state of $ M $ is a column vector on $ \mathbb{R} $ satisfying that the summation of entries is always 1.  On input $ w \in \Sigma^* $, let $v_f$ be the final configuration state after scanning the right end-marker $\$$. Define the \emph{accept probability} of $M$ as
\begin{equation}\label{eq:acc-val}
  f_M(w) = \sum_{e_k \in E_a} \frac{| v_f[k] |}{|v_f|} \in [0,1],
\end{equation}
where each value contributes with its absolute value. More specifically, when $M$ is in the final state $v_f$, this vector is normalized with respect to the $\ell_1$-norm obtaining a new vector $v_f'$; thus, in order to obtain the accept probability we project the vector $v_f'$ on the subspace spanned by the accept inner states $E_a$ of $M$ and then taking the $\ell_1$-norm again, that is, the summation of the absolute value of each entry. 

Language recognition for $M$ is defined in the same way. Any language recognized by an AfA with cutpoint is called an \textit{affine language}. The class of affine languages is denoted $ \affine $. Any language recognized by an AfA with cutpoint 0 (called nondeterministic AfA or NAfA for short) is called a \textit{nondeterministic affine language}. The corresponding class is denoted $ \naffine $. A language is called an \emph{exclusive affine language} if and only if there exists an AfA $M$  and a cutpoint $ \lambda \in [0,1]  $ such that $L = \{ w \in \Sigma^* \mid f_M(w) \neq \lambda \}$. The class of exclusive affine languages is denoted by $ \excaffine $ and its complement class is denoted by $ \coexcaffine $. Any language recognized by an AfA with bounded error is called an \textit{bounded affine language}. The corresponding class is denoted $ \baffine $. If the error is positive one-sided (all non-members are accepted with value 0), then the corresponding language class is denoted $ \posonebaffine$, whereas for negative one-sided error (all members are accepted with value 1) the corresponding language class is denoted $ \negonebaffine $. Note that if $ L \in \posonebaffine $, then $ \overline{L} \in \negonebaffine $, and vice versa. Any language recognized by an AfA with zero-error is called \textit{exact affine language} and its corresponding language class is $ \exactaffine $.

\section{Simulation of Rational PFAs}\label{sec:pfa-sim}

In this section we present a simulation of PFAs by AfAs. Since 1-state PFAs are trivial, we focus on PFAs with two more states.
\begin{theorem}
	\label{the:pfa-sim}
	Any language $ L $ recognized by an $n$-state rational PFA with cutpoint $ \frac{1}{2} $ can be recognized by an $ (n+1) $-state integer AfA with cutpoint $ \frac{1}{2} $.
\end{theorem}
\begin{proof}
	Let $ P=(E,\Sigma,\{ A_\sigma \mid \sigma \in \TSigma \},e_s,E_s) $ be an $n$-state PFA defined with only rational numbers with $ n>1 $. With the help of end-markers, we can assume with no loss of generality that the initial state $ e_s = e_1 $ and $ E_a = \{ e_1 \} $. Moreover, for any given $ w \in \Sigma^* $, we can assume with no loss of generality that the final state vector of $ M $ is always
\[
	\myvector{f_P(w) \\ 1 - f_P(w) \\ 0 \\ \vdots \\ 0}.
\]

Using $ P $ as defined above, we construct an AfA $ MP = (E \cup \{e_{n+1}\},\Sigma,\{ B_\sigma \mid \sigma \in \TSigma \},e_1,\{e_1\}) $, where $n=|E|$. Let $d$ be the smallest positive integer such that for each $\sigma\in\Sigma$ the entries of the matrix $ d A_\sigma $ are integers. If $v_0$ is the initial state of $P$, for any string $ w $, we have that
\[
	\left(d A_{\tildew_{|\tildew|}} \right) \left( d A_{\tildew_{|\tildew|-1}} \right) \cdots \left( d A_{\tildew_{1}} \right) v_0 = d^{|\tildew|} \myvector{f_P(w) \\ 1 - f_P(w) \\ 0 \\ \vdots \\ 0} \in \mathbb{Z}^n.
\]
Define a new matrix $ A'_\sigma $ for each $ \sigma \in \TSigma $ as
\[
	A'_\sigma=\mymatrix{ c|c } { dA_\sigma & 0 \\ \hline \overline{1} & 1 },
\]
where $ \overline{1} $ is a row vector that makes the summation of each column under $dA_\sigma$ equal to 1. Then, for a given string $ w $, we have that
\[
	v'_f = A'_{\tildew_{|\tildew|}} A'_{\tildew_{|\tildew|-1}} \cdots A'_{\tildew_{1}} \myvector{v_0 \\ 0} = \myvector{d^{|\tildew|}  f_P(w) \\ d^{|\tildew|} \left( 1 - f_P(w) \right) \\ 0 \\ \vdots \\ 0 \\ 1 - d^{|\tildew|}} \in \mathbb{Z}^{n+1}.
\]
Using the vector $v_f'$, we can subtract the second entry from the first one and then sum everything else on the second entry by using an extra affine operator $ A''_{\tildew_{|\tildew|}} $ obtaining
\[
	v''_f=A''_{\tildew_{|\tildew|}} v_f' = \myvector{ d^{|\tildew|} \left( 2f_P(w) - 1 \right) \\ 1 - d^{|\tildew|} \left( 2f_P(w) - 1 \right) \\ 0 \\ \vdots \\ 0 } \in \mathbb{Z}^{n+1}.	
\]
Here the entries of $ A''_{\tildew_{|\tildew|}} $ are as follows:
\[
	A''_{\tildew_{|\tildew|}} =
	\mymatrix{rrrrr}{~1 & -1 & ~0 & \cdots & ~0 \\ 0 & 2 & 1 & \cdots & 1 \\ \hline  \multicolumn{5}{c}{ \mathbf{0}} },
\]
where $ \mathbf{0} $ is a $ (n-1,n+1) $-dimensional zero matrix.
The vector $v_f''$ is our desired final state for machine $ MP $. Thus, for each $ \sigma \in \Sigma \cup \{ \cent \} $, we set $ B_\sigma = A'_\sigma $, and, for the last operator we set $ B_\dollar = A''_\dollar A'_\dollar $. The initial vector of $ MP $ is $ u_0 = \myvector{v_0 \\ 0 }$. Then, $ \forall w \in \Sigma^* $, $ f_P(w) > \frac{1}{2} \leftrightarrow f_{MP}(w) > \frac{1}{2} $.
\qed
\end{proof}
The simulation in Theorem \ref{the:pfa-sim} is helpful for recognizing rational exclusive stochastic languages with bounded-error.

\begin{theorem}
	\label{the:sl-bafl}	
	$ \excstocrational \subseteq \baffine_{\mathbb{Z}}^{0} $. 
\end{theorem}
\begin{proof}
	Let $ L \in \excstocrational $ and let $ P $ be a PFA with rational transitions such that if $ w \in L $ then $ f_P(w) \neq \frac{1}{2} $, and if $ w \notin L $ then $ f_P(w) = \frac{1}{2} $.

Based on the simulation of Theorem \ref{the:pfa-sim}, we can construct an AfA $ MP $ that simulates $P$. For non-members of $L$, the first entry of the final vector is always zero and so they are accepted with value 0 by $ MP $; for members of $L$, the first entry of the final vector can be a non-zero integer. Then, the final vector of $ MP $ can be one of the followings
	\[
		\cdots,\	\mymatrix{r}{2 \\-1 \\ 0 \\ \vdots \\ 0}, \myvector{1\\0\\ 0 \\ \vdots \\ 0}, \mymatrix{r}{-1\\2\\ 0 \\ \vdots \\ 0}, \mymatrix{r}{-2\\3\\ 0 \\ \vdots \\ 0},\cdots .
	\]
	Hence, the accepting value is at least $ \frac{1}{3} $ for members of $L$. By executing a few copies of $ MP $ in parallel, we can increase the accept value arbitrary close to 1. Considering that all non-members of $L$ are accepted with zero value, the recognition mode is positive one-sided bounded-error.
	\qed
\end{proof}

The following corollary is obtained immediately from Theorem \ref{the:sl-bafl}.

\begin{corollary}
	$  \coexcstocrational \subseteq \baffine_{\mathbb{Z}}^{1} $.
\end{corollary}

It was shown in \cite{DCY16A} that $ \excstoc = \naffine = \nqal $, and therefore, our new result is stronger (bounded-error) but for a restricted case (using only rational numbers). One may immediately ask whether $ \baffine_{\mathbb{Q}}^{0} \subseteq \excstocrational $. This follows from the simulation of a NAfA by a NQFA given in \cite{DCY16A}, and the a simulation of a NQFA by PFA with exclusive cutpoint (see \cite{YS10A}). Note that all the intermediate machines can use only rational transitions. Moreover, we can give a direct simulation of a NAfA by a PFA by using Turakainen's techniques \cite{Tur75}. 

\begin{corollary}
	$ \baffine_{\mathbb{Z}}^{0} = \baffine_{\mathbb{Q}}^{ 0} = \excstocrational $ and $ \baffine_{\mathbb{Z}}^{1} = \baffine_{\mathbb{Q}}^{1} = \coexcstocrational $.
\end{corollary}

The class $ \excstocrational $ is important because, as pointed in \cite{YS10B}, it contains many well-known nonregular languages like $ \mathtt{UPAL} = \{ a^nb^n \mid n > 0 \} $, $ \mathtt{PAL} = \{ w \in \Sigma^* \mid w = w^r \} $, $ \mathtt{SQUARE} = \{ a^n b^{n^2} \mid n >0 \} $, $ \mathtt{POWER} = \{ a^n b^{2^n} \mid n >0 \} $, etc. Interestingly, any language in $ \excstocrational $ ($ \coexcstocrational $) can also be recognized by two-way QFAs with positive (one-sided) bounded-error. Therefore, it is reasonable to compare AfAs with two-way QFAs.

We can provide logarithmic-space bounds for one-sided bounded-error affine languages. We know that $ \excstocrational \cup \coexcstocrational $ is in the deterministic logarithmic space class $L$ \cite{Mac98} and $ \mathtt{PAL} $ cannot be recognized by a probabilistic Turing machine in sublogarithmic space \cite{FK94}. Hence, we can immediately obtain the following result.

\begin{corollary}
	 $ \baffine_{\mathbb{Q}}^{0} \cup \baffine_{\mathbb{Q}}^{1} \subseteq \mathsf{L} $ and
	 $ \baffine_{\mathbb{Q}}^{0} \cup \baffine_{\mathbb{Q}}^{1} \nsubseteq\mathsf{BSpace(o(\log n))} $.
\end{corollary}

The language $ \mathtt{EQNEQ} = \{ aw_1 \cup bw_2 \in \{a,b\}^* \mid w_1 \in \mathtt{EQ} \mbox{ and } w_2 \in \mathtt{NEQ} \}  $ is not in $ \excstoc \cup \coexcstoc $, where $ \mathtt{EQ} = \{ w \in \{a,b\}^* \mid |w|_a = |w|_b \} $ and $ \mathtt{NEQ} $ is the complement of $ \mathtt{EQ} $ \cite{YS10A}. We know that $ \mathtt{EQ} $ can be recognized by an AfA with bounded-error, and hence, it is not hard to design an AfA recognizing $ \mathtt{EQNEQ} $ with bounded-error; the error, however, must be two-sided since it is not in $ \excstoc \cup \coexcstoc $.

\begin{theorem}
	$ \baffine_{\mathbb{Q}}^{0} \cup \baffine_{\mathbb{Q}}^{1} \subsetneq \baffine_\mathbb{Q} $.
\end{theorem}

\section{Exact Simulation of QFAs}\label{sec:mcqfa-sim}

In this section, we present an exact simulation of QFAs by AfAs. We start with the exact simulation of MCQFAs due to its simplicity.

\begin{lemma}
	\label{lem:mcqfa-sim}
	For a given MCQFA $M$ with $n$ inner states defined over $ \mathbb{R}^n $, there exists an AfA $MM$ with $(n^2+1)$ inner states that exactly simulates $M$.
\end{lemma}
\begin{proof}
Let $ M = (Q,\Sigma,\{ U_\sigma \mid \sigma \in \TSigma \},q_s,Q_a) $ be an $ n $-state MCQFA and $ \ket{v_0} = \ket{q_s} $ be the initial quantum state. All transitions of $M $ use only real numbers. For any given input $ w \in \Sigma^* $, the final quantum state is $ \ket{v_f} $ is
\[
	\ket{v_f} = U_{\tildew_{|\tildew|}} U_{\tildew_{|\tildew|-1}} \cdots U_{\tildew_{1}} \ket{v_0}.
\]
In order to turn amplitudes into probabilities of observing the basis states from the final vector, we can tensor $\ket{v_f}$ with itself \cite{MC00}. Thus,
\[
	\ket{v_f} \otimes \ket{v_f} = (U_{\tildew_{|\tildew|}} \otimes U_{\tildew_{|\tildew|}}) (U_{\tildew_{|\tildew|-1}} \otimes U_{\tildew_{|\tildew|-1}}) \cdots (U_{\tildew_{1}} \otimes U_{\tildew_{1}} ) (\ket{v_0} \otimes \ket{v_0}).
\]

We construct an AfA $ MM $ that simulates the computation of $ M $. 
The set of inner states is $  Q \times Q \cup\{ q_{n^2+1} \}  $ and the initial state is $ (q_s,q_s) $. We assume with no loss of generality that there is only one accept state $ (q_1,q_1)$. For any symbol $ \sigma \in \Sigma \cup \{ \cent \} $, the transition affine matrix $ A_\sigma $ is defined as
\[
	A_\sigma = \mymatrix{c|c}{ U_\sigma \otimes U_\sigma & 0 \\ \hline \overline{1} & 1 },
\]
where $ \overline{1} $ is a row vector that makes the summation of each column under $U_\sigma \otimes U_\sigma$ equal to 1. The affine transformation $ A_\dollar $ is composed by two affine operators
\[
	A_\dollar = A'_\dollar \mymatrix{c|c}{ U_\sigma \otimes U_\sigma & 0 \\ \hline \overline{1} & 1 },
\]
where $ A'_\dollar $ is an affine operator to be specified later. Then, on input $ w $, the final affine state is 
\[
	u_f = A'_\dollar \myvector{ v_f \otimes v_f \\ \overline{1} },
\]
where $ \overline{1} $ is equal to 1 minus the summation of the rest of the entries in $u_f$. The accept value of $M$ on $w$ can now be calculated from the values of $ v_f \otimes v_f $, that is, the summation of entries corresponding to $ (q_j,q_j) $ for all $ q_j \in Q_a $. Similar to the simulation in the previous section, we define $ A'_\dollar $ as an operation that 
computes the  summation over all entries corresponding to each accepting state of the form $(q_j,q_j)$ and copies the result to the first entry of $u_f$; all remaining values are added and copied to the second entry of $u_f$. (The first and second rows of $ A'_\dollar $ are 0-1 vectors and all the other rows are by zero vectors.) Thus, our final state is $u_f=(f_M(w),1-f_M(w),0,\dots,0)^T$.
Finally, we have that $ f_{MM}(w) $ equals $ f_M(w) $ and the number of inner states of $MM$ is $ n^2 +1 $.\qed
\end{proof}

It is known that the computation of any $ n $-state QFA $M$ (defined with complex numbers) can be simulated by an $n^2$-state \emph{general finite-state automaton} $G$ such that $ f_M(w) = f_G(w) $ for any $ w \in \Sigma^* $ \cite{YS11A}. Then, by adding one more state, we can design an AfA $ MM $ such that $ f_M(w) = f_{MM}(w) $ for any $ w \in \Sigma^* $. Hence, the following result immediately follows.

\begin{theorem}
	\label{the:qfa-sim}
	For a given QFA $M$ with $n$ inner states, there exists an AfA $MM$ with $(n^2+1)$ inner states that exactly simulates $M$.
\end{theorem}

By using this theorem, we inherit the superiority results of QFAs over PFAs \cite{AY15} as the superiority results of AfAs over PFAs. The only issue we should be careful about is the quadratic increase in the number of states, which could be significant depending on the context.

The simulation techniques given here can be applied to different cases. For example, an affine circuit can be defined similarly to a quantum circuit, using affine operators instead of unitary operators. Then, using the above simulation(s), it follows that any quantum circuit of width $ d(n) $ and length $ s(n)$ can be simulated exactly by an affine circuit of width $ d^2(n)+1 $ and length $ s(n) $, where $ n $ is the parameter of the input length. Therefore, we can say that the class $ \sf BQP $ is a subset of bounded-error affine polynomial-time defined with circuits. Moreover, $\sf  PSPACE $ is a trivial upper bound for these polynomial-time  circuits.

\section{Succinctness of Affine Computation.}\label{sec:succ}
\subsection{Bounded-error}\label{sec:succ-be}

For any prime number $p$, the language $ \mathtt{MOD_p} = \{ a^{jp} \mid j \geq 0 \} $, over the unary alphabet $\{a\}$, can be recognized by a bounded-error MCQFA with $ O(\log(p)) $ inner states; any bounded-error PFA, however, requires at least $ p $ states \cite{AF98}. The MCQFA algorithm for $ \mathtt{MOD_p} $ is indeed composed by $ O(\log(p)) $ copies of 2-state MCQFAs. Since we can simulate these 2-state MCQFAs exactly by 5-state AfAs, it follows that $\mathtt{MOD_p} $ can be recognized by bounded-error AfAs with $ O(\log(p)) $ inner states. 
	
	The language $ \mathtt{COUNT_n} = \{ a^n \} $ for any $ n>0 $, also known as the (unary) \emph{counting problem}, can be solved by bounded-error AfAs with a constant number of states; moreover, any DFA requires $ n $ states \cite{KTSV99}, which implies that any bounded-error QFAs must have at least $\Omega(\sqrt{\log(n)}) $ states \cite{AY15}.
	
	\begin{theorem}
		\label{the:count}
		The language $ \mathtt{COUNT_n} $ can be recognized by a 2-state AfA with negative one-sided bounded-error.
	\end{theorem}
	 \begin{proof}
	 	We use two states $ \{e_1,e_2\} $ where $ e_1 $ is the initial and only accept inner state. Over an alphabet $ \Sigma = \{a\} $, we define the initial configuration state $v_0$ and $A_{\cent}$ as \[
	v_0 = \myvector{1 \\ 0} \text{ and }
	A_{\cent} = \mymatrix{cc}{ 2^n & ~~0 \\ 1-2^n & ~~1 },
\] respectively.  After scanning the left end-marker, the configuration state is
	\[
		v_1 = \myvector{ 2^n \\ 1-2^n }.
	\]
	For each symbol $a$, we apply the operation
\[
 A_a=\mymatrix{cc}{ \frac{1}{2} & 0 \\ \frac{1}{2} & 1 }.
\]
Then, after reading $ j $ symbols, we have that
	\[
		v_{j+1} = \myvector{ 2^{n-j} \\ 1 - 2^{n-j} }.
	\]  
	Finally, we define $A_\$$ as the identity operation. For the single member of $\mathtt{MOD_p}$, namely the string $ a^n $, the final configuration state is
	\[
		v_f = \myvector{ 1 \\ 0 },
	\]
	and hence, it is accepted exactly. For any non-member string of $\mathtt{MOD_p}$, the final configuration state can be one of the followings
	\setstretch{1.3}
	\[
		\cdots,\mymatrix{r}{ 4 \\ -3 }, \mymatrix{r}{ 2 \\ -1 }, \myvector{ \frac{1}{2} \\ \frac{1}{2} },\myvector{ \frac{1}{4} \\ \frac{3}{4} }, \cdots
	\]
	\setstretch{1}
	and, in consequence, the accept value can be at most $ \frac{2}{3} $.
	 \qed
	 \end{proof}
	Using a few copies of the AfA of Theorem \ref{the:count}, the error can be made arbitrarily close to 0 with a number of inner states that depends only on the error bound.

\subsection{Zero-error}
For any $k > 0 $, $ \modtwo  = (\modtwoyes,\modtwono) $ is a promise problem,\footnote{A promise problem $\tt L=(L_{yes},L_{no})$ is solved by a machine $M$, or $M$ \emph{solves} $\tt L$, if for all $ w \in \tt L_{yes}$, $M$ accepts $w$, and for all $ w \in \tt L_{no}$, $M$ rejects $w$.} where
$ \modtwoyes = \{ a^{j2^k} \mid j \equiv 0 \mod 2 \} $ and
$ \modtwono = \{ a^{j2^k} \mid j \equiv 1 \mod 2 \} $.

	It is known that $\modtwo$ can be solved exactly by a 2-state MCQFA \cite{AY12}. Any bounded-error PFA, however, requires at least $ 2^{k+1} $ states \cite{RY14A}. Due to Lemma \ref{lem:mcqfa-sim}, we can obtain the following result.
	\begin{theorem}
		The promise problem $ \modtwo $ can be solved by a 5-state AfA exactly.
	\end{theorem}
	In consequence, zero-error AfAs are also interesting like MCQFAs. Now
	consider the promise problem $ \modfour=(\modfouryes,\modfourno) $ where  $ \modfouryes = \{ a^{j.2^k} \mid j \equiv 0 \mod 4   \} $ and  $ \modfourno = \{ a^{j.2^k} \mid j \equiv 1 \mod 4   \} $.
	\begin{theorem}
	\label{the:mod-afa}
		The promise problem $ \modfour $ can be solved exactly by a 3-state AfA.
	\end{theorem}
	\begin{proof}
		We use the algorithm given in \cite{AY12}, but there is no need to tensor it with itself to solve $ \modfour $. Let $M_k$ be a MCQFA with two inner states defined with real numbers; moreover, $M_k$ does not need to use end-markers. The initial configuration state of $M_k$ is
\[
v_0=\myvector{1 \\ 0}.
\]
After reading blocks of size $ |a^{2^k}| $, the configuration states of $M_k$ change as follows:
		\[
			v_0=\myvector{1 \\ 0} \xrightarrow{a^{2^k}}
			\myvector{0 \\ 1} \xrightarrow{a^{2^k}}
			\myvector{-1 \\ 0} \xrightarrow{a^{2^k}}
			\myvector{0 \\ -1} \xrightarrow{a^{2^k}}
			\myvector{1 \\ 0} \xrightarrow{a^{2^k}}
			\myvector{0 \\ 1} \xrightarrow{a^{2^k}}
			\cdots
		\]
		We can simulate this computation using a 3-state AfA; hence, an affine configuration state becomes
		\[
			\myvector{1 \\ 0 \\ 0} \text{ and } \myvector{0 \\ 1 \\ 0}
		\]
		after reading the $ 4j $-th and $ (4j+1) $-th blocks for $ j \geq 0 $. This suffices to exactly solve $ \modfour $.
	\qed
	\end{proof}
Using the techniques given in \cite{AY12,RY14A}, we can show that any bounded-error PFA (and some other classical automata models \cite{GefY15A}) requires at least $ 2^{k+1} $ states to solve $ \modfour $. 
	
In summary, we can say that $ \modtwo $ (and so $ \modfour $) is a classically expensive promise problem, but inexpensive for quantum and affine automata. As further examples, in the same line of research, a classically expensive generalized version of $ \modtwo $  was defined in \cite{GQZ15}, in which was shown that the same expensive language can be solved by 3-state MCQFAs exactly; furthermore, a classically expensive function version of $ \modtwo $ was defined in \cite{AGKY14A}, which was shown to be solved by width-2 quantum OBDDs exactly. Trivially, all quantum results for these families of promise problems are inherited  for affine models.

\section{Unary Languages Recognized by Affine Automata with Two Inner States}\label{sec:unary}

All of our results of the previous sections, excepting the succintness results of Section \ref{sec:succ}, were obtained for languages defined over generic alphabets. Hence, using the superiority result of QFAs over PFAs given in \cite{GaiY15A}, it immediately follows that AfAs computing unary languages are more powerful than unary PFAs with bounded-error on promise problems.

In this section, we give a complete characterization of the unary languages recognized by 2-state AfAs with cutpoint. It is known that 2-state unary PFAs can recognize only a few regular languages, whereas 2-state unary QFAs (with transitions defined over $ \mathbb{R}$) can recognize uncountable many languages \cite{Paz71,SY16A}. Here we obtain an analogous result to PFAs with the difference that AfAs can recognize more regular languages.

Consider the following unary regular languages over $ \Sigma = \{a\} $; the empty language $ \emptylang $,
$ \E =\{a\}^*$, $ \less{n} = \{ a^i \mid i \leq n \} $ for $ n \geq 0 $, and $ \even = (aa)^* $.
 
The complete list of languages recognized by 2-state unary PFAs with cutpoint are 
	$ \E $,
	$ \less{n} $, 
	$ \less{n} \cap \even $, 
	$ \less{n} \cap \overline{\even} $, 
	$ \overline{\less{n}} \cap \even $, 
	$ \overline{\less{n}} \cap \overline{\even} $,
	and the complement of each of these languages, with $ n \geq 0 $ \cite{SY16A}.

The main result of this section is the following. Let $\interval{k}{l} = \{ a^i \mid k \leq i \leq l \}$ for $ 1 \leq k < l $.

\begin{theorem}
\label{the:cha}
The only unary regular languages recognized by AfAs with 2 inner states are the languages recognized by 2-state unary PFAs with cutpoint and additionally $ \interval{k}{l} \cap \even $, 
	$ \interval{k}{l} \cap \overline{\even} $, 
	$ \overline{\interval{k}{l}} \cap \even $, and
	$ \overline{\interval{k}{l}} \cap \overline{\even} $.
\end{theorem}

The remaining of this section is devoted to the proof of Theorem \ref{the:cha}. To that end, first we will consider the computation of a 2-state unary AfA $ M $, which is inspired by a 2-state unary PFA  of \cite{SY16A}. Let $ \{e_1,e_2\} $ be the only inner states of $ M $. With no loss of generality, we assume that the initial and only accepting state is $ e_1 $. The affine transformations for symbols $ a $ and $\dollar$  are
\begin{eqnarray}
	A_a &=& \mymatrix{ccc}{ 1 - q & ~~ & p \\ q & & 1-p  }\text{ and }\label{eq:aa}\\
	A_\dollar &=& \mymatrix{ccc}{ f_1 & ~~ & f_2 \\ 1-f_1 & & 1-f_2  }\label{eq:adollar},
\end{eqnarray}
respectively, for some real numbers $ p,q,f_1$ and $f_2$. 

Let $ v_f = \myvector{x \\ 1-x} $ be the final configuration vector of string $ a^j $ ($ j \geq 0 $). The accept probability of $ M $ on $ a^j $ is $f_M(a^j) = \dfrac{-x}{1-2x} =  \dfrac{1}{2} + \dfrac{1}{4x-2}$ when $x<0$ and $x>1$, and $f_M(a^j) = x$ when $0\leq x \leq 1$.

\begin{lemma}
If $ p=q=0 $ in Eq.(\ref{eq:aa}), then $ \E $ and $ \emptylang $ can be recognized by AfAs with 2 states.
\end{lemma}
\begin{proof}
It is clear that if $ p=q=0 $, then $ A_a $ is the identity, and hence $ f_M $ is a constant function on $ \Sigma^* $. Thus, $ M $ can recognize $ \E $ and $ \emptylang $.\qed
\end{proof}

For the remaining of this section, we assume that at least one of $p$ or $q$ is non-zero.

\begin{lemma}
\label{lem:lessn}
There exists $p\in \mathbb{R}$ satisfying $ p+q = 0 $ in Eq.(\ref{eq:aa}) such that $M$ recognizes $ \less{n}$.
\end{lemma}
\begin{proof}
	Suppose that $ p+q = 0 $ in $ A_a $. Then, we have
\[
	A_a = \mymatrix{ccc}{ 1+p & ~~ & p \\ -p & & 1-p }.
\]
If the initial state is $ \myvector{m \\ 1-m} $, then we can obtain:
\[
	\myvector{m+p \\ 1-m-p} = \mymatrix{ccc}{ 1+p & ~~ & p \\ -p & & 1-p } \myvector{m \\ 1-m}
\]
and then
\[
	\myvector{m+jp \\ 1-m-jp} = \mymatrix{ccc}{ 1+p & ~~ & p \\ -p & & 1-p }^j \myvector{m \\ 1-m}
\]
for $ j > 1 $. The accept probability of $M$ can be seen in Figure \ref{fig:accept-prob}.

If we start at point $ m $ in Figure \ref{fig:accept-prob}, we shift either left or right with amount of $ |p| > 0 $ for each reading symbol. In this way, it is easy to see that $ M $ can recognize $ \less{n} $ for any $ n \geq 0 $.
\begin{figure}[t]
\centering
\includegraphics[scale=0.5]{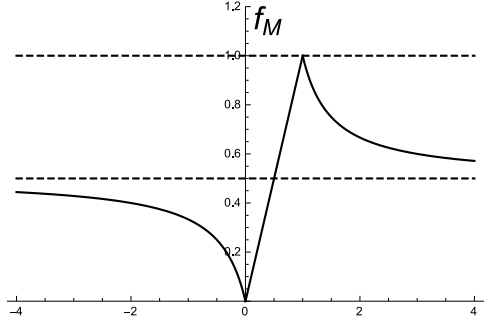}
\caption{Accept probability of $M$.}
\label{fig:accept-prob}
\end{figure}	
\qed
\end{proof}

\begin{corollary}
There exists $p\in \mathbb{R}$ satisfying $ p+q = 0 $ in Eq.(\ref{eq:aa}) such that $M$ recognizes $\overline{ \less{n} }$.
\end{corollary}

\begin{lemma}
	\label{lem:interval}
	There exists  $p\in \mathbb{R}$ satisfying $p+q=0$ such that $M$ recognizes $\interval{k}{l}$ with cutpoint $3/4$.
\end{lemma}
\begin{proof}
	Let $ l-k = d > 0 $ and the cutpoint be $ \frac{3}{4} $. We pick $ m $ as $ \frac{3}{2} + \frac{3}{4d+8}(k-1) $ and starts in state $ \myvector{m \\ 1-m} $ after reading the left end-marker. We apply $ A_a $ for each symbol where $ p = - \frac{3}{4d+8} $. That means, on the number line in Figure \ref{fig:accept-prob}, we start at  $ m $ and then shift to the left with value $ \frac{3}{4d+8} $ for each scanned symbol. We pick $ A_\dollar $ as identity. After reading $ (k-1) $ symbols, the first entry of $ v_f $ will be
\[
	\frac{3}{2} + \frac{3}{4d+8}(k-1) - (k-1) \frac{3}{4d+8} = \frac{3}{2}.
\]
Then, $ f_M(a^{k-1}) = \frac{|\frac{3}{2}|}{|\frac{3}{2}|+|-\frac{1}{2}|} \frac{3}{4} $, and so $ f_M(a^j) \leq \frac{3}{4} $ for any $ j < k $. After reading one more symbol, the accepting probability ($ f_M(a^k) $) becomes more than $ \frac{3}{4} $. As can be easily seen from the graph, by continuing to read input symbols, the accepting probability may stay over $ \frac{3}{4} $ for a while, and then it always stays below $ \frac{3}{4} $. More specifically, if we read $ (d+2) $ more symbols (after reading $ (k-1) $ symbols), the first entry of $ v_f $ will be
\[
	\frac{3}{2} - (d+2) \frac{3}{4d+8} = \frac{3}{2} - \frac{3}{4} = \frac{3}{4}.
\]
So, the accepting probability hits $ \frac{3}{4} $ again for the string $ a^{k-1+d+2} = a^{k+l-k+1} = a^{l+1} $. Therefore, the accepting probabilities of all strings
\[
	a^k, a^{k+1},\ldots,a^l
\]
are more than $ \frac{3}{4} $, and the accepting probabilities of all other strings are at most $ \frac{3}{4} $.
\qed
\end{proof}

\begin{corollary}\label{cor:not-interval}
The language $ \overline{\interval{k}{l}}$ can be recognized by a 2-state AfA with cutpoint $ \frac{3}{4} $.
\end{corollary}

From Lemma \ref{lem:interval} and Corollary \ref{cor:not-interval} we conclude that AfAs with two inner states can recognize more languages than PFAs with two inner states. Moreover, for the case of $p+q=0$, there are no more regular unary languages recognized by AfAs with two states (see Appendix \ref{sec:pq}).

\begin{lemma}
\label{lem:char}
There exists $p,q\in\mathbb{R}$ satisfying $p+q\neq 0$ such that $M$ recognizes all languages recognized by 2-state unary PFAs with cutpoint and the languages
	$ \interval{k}{l} \cap \even $,
	$ \interval{k}{l} \cap \overline{\even} $, 
	$ \overline{\interval{k}{l}} \cap \even $, and
	$ \overline{\interval{k}{l}} \cap \overline{\even} $.
\end{lemma}
\begin{proof}
	First, we identify a fix point of $ A_a $:
\[
	\myvector{ \frac{p}{p+q} \\ \\ \frac{q}{p+q} } =  \mymatrix{ccc}{ 1-q & ~~ & p \\ \\ q & & 1-p } \myvector{ \frac{p}{p+q} \\ \\ \frac{q}{p+q} }.
\]
Then, we can assume that after reading the left end-marker, we are in state
\[
	\myvector{ \frac{p}{p+q} + c \\ \\ \frac{q}{p+q} -c } 
\]
for some real $ c $. After applying $ A_a $, we obtain
\[
	\myvector{ \frac{p}{p+q} + c (1-p-q) \\ \\ \frac{q}{p+q} -c (1-p-q) } = \mymatrix{ccc}{ 1-q & ~~ & p \\ \\ q & & 1-p }  \myvector{ \frac{p}{p+q} + c \\ \\ \frac{q}{p+q} -c },
\]
and after applying $ A_a^j $
\[
	\myvector{ \frac{p}{p+q} + c (1-p-q)^j \\ \\ \frac{q}{p+q} -c (1-p-q)^j } = \mymatrix{ccc}{ 1-q & ~~ & p \\ \\ q & & 1-p }^j  \myvector{ \frac{p}{p+q} + c \\ \\ \frac{q}{p+q} -c }
\]
for $ j >0 $.
Let $ \frac{p}{p+q} = r $ and $ (1-p-q) = t \neq 1 $ (remember that we assume $ p+q \neq 0 $). Then
\[
	\myvector{ r+ct^j \\ 1- r- ct^j } = A_a^j \myvector{r \\ 1 -r}.
\]
After reading the right end-marker, the final state will be
\[
	v_f = \myvector{ f_1 r - f_2 r + f_2 + (f_1-f_2)ct^j \\ \\ \overline{1} } =  \mymatrix{ccc}{ f_1 & ~~ & f_2 \\ \\ 1-f_1 & & 1-f_2  } \myvector{ r+ct^j \\ \\ 1- r- ct^j },
\]
where $ \overline{1} = 1 - v_f[1] $. Let $ F = f_1 r - f_2 r + f_2 $ and $ C = (f_1-f_2)c $. Then, the first entry of $ v_f $ becomes
\[
	E = F + Ct^j,
\]
which determines the accepting probability. For the empty string we have $ E = F+C $. Another trivial case is $ C = 0 $ where we have a fixed accept probability for all strings. Assume that $ C \neq 0 $ in the following.  
\begin{itemize}
	\item If $ 0 < t < 1 $, then, $ E $ converges to $ F $ from either left or right, depending on the sign of $ C $.
	\item Remark that $ t \neq 1 $.
	\item If $ t > 1 $, then, $ E $ diverges as $ F+C, F+tC, F+t^2C, \ldots $.
	\item If $ -1 < t < 0 $, then $ E $ converges from both side.
	\item If $ t = -1 $, then $ E=F+C $ for the strings with even length and $ E=F-C $ for the strings with odd length.
	\item If $ t < -1 $, then $ E $ diverges as $ F+C, F-tC, F+t^2C,F-t^3C, F+t^4C,\ldots $.
\end{itemize}
A careful analysis of the graph of Figure \ref{fig:accept-prob} show that $ M $ recognizes all languages recognized by 2-state unary PFAs with cutpoint and the followings:
	$ \interval{k}{l} \cap \even $, 
	$ \interval{k}{l} \cap \overline{\even} $, 
	$ \overline{\interval{k}{l}} \cap \even $, and
	$ \overline{\interval{k}{l}} \cap \overline{\even} $.
\qed
\end{proof}

With Lemma \ref{lem:char} we conclude the proof of Theorem \ref{the:cha}.

\section{Concluding Remarks}

Affine computation and affine finite automata were introduced in \cite{DCY16A} with a few initial results. For example, it was proved that AfAs can recognize more languages than PFAs and QFAs in the bounded and unbounded error modes, the exclusive affine languages form a superset of the exclusive quantum and stochastic languages, and nondeterministic AfAs and QFAs are both equivalent to the class of exclusive stochastic languages. 

In this paper, we continued to investigate AfAs and obtained some new and complementary results. We presented efficient simulations of PFAs and QFAs by AfAs. In addition, we characterized the class of languages recognized by positive and negative one-sided bounded-error AfAs using rational transitions, which turn out to be equal to the union of rational exclusive and co-exclusive stochastic languages; this latter result improved the proof of equivalence between nondeterministic AfAs and QFAs. We also initiated the study of the state complexity of AfAs and showed that they can be more succint than PFAs and QFAs. Finally, we presented a complete characterization of 2-state unary AfAs, showing at the same time that AfAs can recognize more languages than 2-state unary PFAs  but still only regular languages.

In a recent and related work on AfAs \cite{BMY16A}, some further results on state complexity are presented. In \cite{BMY16A}, it was proved that AfAs can separate any pair of strings with zero-error using only two states, and can separate efficiently any pair of disjoint finite sets of words with one-sided bounded-error.

We close this paper with a few open problems that we consider challenging.
\begin{enumerate}
	\item It is conjectured in \cite{DCY16A} that affine and quantum computation can be incomparable. The simulation results in this paper give the feeling that quantum models can be simulated by their affine counterparts but it might require a quadratic increase in memory. It is interesting to study the relations, particularly in the bounded-error setting, between quantum and affine language classes.
	\item Currently we are not aware of any non-trivial upper bound for $ \baffine_\mathbb{Q} $. Using the techniques of \cite{Mac98} it might be possible to prove an upper bound of logarithmic space.
	\item Considering that AfAs completely capture the power of NQFAs, it is interesting to investigate lower bound techniques that can exploit the simpler structure of affine transformations (compared to unitary and positive-semidefinite operators).
\end{enumerate}

\section*{Acknowledgement} 
We thank to the anonymous referees for their helpful comments.

\bibliographystyle{splncs03}
\bibliography{tcs}

\newpage
\appendix

\section{All Regular Unary Languages Recognized by Two-state AfAs in the Case $p+q=0$}\label{sec:pq}
On input $a^j$, the final configuration state is
\[
v_f=\begin{pmatrix}
f_1	& f_2\\
1-f_1 ~	& 1-f_2
\end{pmatrix}
\begin{pmatrix}
m+j\cdot p\\
1-m-j\cdot p
\end{pmatrix}
=
\begin{pmatrix}
m f_1+j p f_1+f_2-m f_2-j p f_2\\
\bar{1}
\end{pmatrix}
\]
where $\bar{1} = 1-(m f_1+j p f_1+f_2-m f_2-j p f_2)$. If we let $F=m(f_1-f_2)+f_2$ and $C=pf_1-pf_2$, then we can write
\[
v_f=\begin{pmatrix}
m(f_1-f_2)+f_2+j(pf_1-pf_2)\\
\bar{1}
\end{pmatrix}
=
\begin{pmatrix}
F+jC\\
\bar{1}
\end{pmatrix}
\]
For convenience, we let $E=F+jC$.

If $C=0$, then $f_M$ is a constant function and $M$ recognizes either $E$ or $\emptylang$. From now on we assume that $C\neq 0$.

Let the cutpoint $\lambda<1/2$. 
\begin{itemize}
	\item If $E<\frac{\lambda}{2\lambda-1}$ and $C$ is negative, then $f_M(x)>\lambda$ for all $x$, and hence, $M$ recognizes $\E$; if $C$ is positive, then $M$ recognizes $\overline{\interval{k}{l}}$ for some $k>0$ and $l\geq 0$.
	\item If $E=\frac{\lambda}{2\lambda-1}$ and $C$ is negative, then $M$ accepts all strings except the empty string; if $C$ is positive, then $M$ accepts all strings except the first $j$ strings, for some $j\geq 0$, and thus, $M$ recognizes $\overline{\less{n}}$.
	\item If $\frac{\lambda}{2\lambda-1} <E<\lambda$, then $M$ recognizes $\overline{\less{n}}$ independent of $C$. 
	\item If $E=\lambda$ and $C$ is negative, then $M$ recognizes $\overline{\less{n}}$; if $C$ is positive, then $M$ accepts all strings except the empty string.
	\item If $E>\lambda$ and $C$ is negative, then the first $k$ strings are not in the language and also all strings after the $j$-th string, hence, $M$ recognizes $\overline{\interval{k}{l}}$; if $C$ is positive, then $M$ recognizes $\E$.
\end{itemize}
   
Now let the cutpoint $\lambda=1/2$. 
\begin{itemize}
	\item If $E\leq \lambda$ and $C$ is negative, then $M$ recognizes $\emptylang$; if $C$ is positive, then $M$ recognizes all strings except the empty string.
	\item If $E>\lambda$ and $C$ is negative, then $M$ recognizes $\less{n}$; if $C$ is positive, then $M$ recognizes $\E$.
\end{itemize}

For the last case, we assume that $\lambda >1/2$. 
\begin{itemize}
	\item If $E\leq \lambda$ and $C$ is negative, then $M$ recognizes $\emptylang$; if $C$ is positive, $M$ recognizes $\interval{k}{l}$.
	\item If $\lambda < E <\frac{\lambda}{2\lambda-2}$, then $M$ recognizes $\less{j}$ independent of $C$. If $E=\frac{\lambda}{2\lambda-2}$ and $C$ is negative, then  $M$ skips the empty string, accepts the first $j$ strings and omits the rest; if $C$ is positive, then $M$ recognizes $\emptylang$.
	\item If $E>\frac{\lambda}{2\lambda-2}$ and $C$ is negative, then $M$ recognizes $\interval{k}{l}$; if $C$ is positive, then $M$ recognizes $\emptylang$.
\end{itemize}

\end{document}